\documentclass[11pt]{article}

\usepackage{amsmath}
\usepackage{amsthm}
\usepackage{amssymb}
\usepackage{hyperref}
\usepackage{pstricks}
\usepackage{algorithm,algpseudocode}
\usepackage{graphicx}
\usepackage{color}

\newtheorem{lemma}{Lemma}[section]
\newtheorem{theorem}[lemma]{Theorem}

\newtheorem{definition}[lemma]{Definition}

\begin{document}

\title{A New Approximation Algorithm for Minimum-Weight $(1,m)$--Connected Dominating Set}
\author{\footnotesize Jiao Zhou$^1$
        \quad Yingli Ran$^1$
        \quad Panos M. Pardalos$^2$
        \quad Zhao Zhang$^1$\thanks{Corresponding Authors: Z. Zhang, hxhzz@sina.com}
        \quad Shaojie Tang$^3$
        \quad Ding-Zhu Du$^4$\\ \\
    {\it\small $^1$ School of Mathematical Sciences, Zhejiang Normal University}\\
    {\it\small Jinhua, Zhejiang, 321004, China}\\
    {\it\small $^2$ Department of Industrial and Systems Engineering, University of Florida}\\
    {\it\small Gainesville, Florida 32611, USA}\\
    {\it\small $^3$ Naveen Jindal School of Management, University of Texas at Dallas}\\
    {\it\small Richardson, Texas 75080, USA}\\
    {\it\small $^4$ Department of Computer Sciences, University of Texas at Dallas}\\
    {\it\small Richardson, Texas 75080, USA}}
\date{}
\maketitle

\begin{abstract}
\noindent Consider a graph with nonnegative node weight. A vertex subset is called a CDS (connected dominating set) if every other node has at least one neighbor in the subset and the subset induces a connected subgraph. Furthermore, if every other node has at least $m$ neighbors in the subset, then the node subset is called a $(1,m)$CDS. The minimum-weight $(1,m)$CDS problem aims at finding a $(1,m)$CDS with minimum total node weight.
In this paper,  we present a new polynomial-time approximation algorithm for this problem
with approximation ratio $2H(\delta_{\max}+m-1)$, where $\delta_{\max}$ is the maximum degree of the given graph and $H(\cdot)$ is the Harmonic function, i.e., 
$H(k)=\sum_{i=1}^k \frac{1}{i}$.
\end{abstract}

\section{Introduction}
For a graph $G=(V,E)$, where $V$ is the node set and $E$ is the edge set, a node subset $C$ is a  {\em dominating set} (DS) of $G$ if any $v\in V\setminus C$ has at least one neighbor in $C$. A dominating set $C$ is a {\em connected DS} (CDS) of $G$ if $G[C]$ is connected, where $G[C]$ is the subgraph of $G$ induced by $C$. The nodes in $C$ are called {\em dominators}, and those in $V\setminus C$ are called {\em dominatees}. The {\em minimum CDS} (MinCDS) problem aims to find a CDS with the minimum cardinality/weight.

MinCDS has wide applications in many fields, including computer science, engineering, and operations research. For example, in a wireless sensor network (WSN), CDSs serve as virtual backbones \cite{Das,Ephremides}, they can save energy and reduce interference while maintaining information sharing.

The sensors in a WSN are prone to failures due to accidental damage or battery depletion. Therefore a fault-tolerant virtual backbone should be maintained. The {\em minimum $k$-connected $m$-fold CDS} (Min$(k,m)$CDS) problem was proposed for this purpose \cite{Dai}. A node subset $C$ is a $(k,m)$CDS if every node in $V\setminus C$ has at least $m$ neighbors in $C$ and the induced subgraph $G[C]$ is $k$-connected.

MinCDSs have been extensively studied, especially in unit disk graphs, UDGs, a widely adopted model of homogeneous WSNs. When nodes have nonnegative weights, the node-weighted versions, namely, minimum weight CDSs (MinWCDSs), have also achieved significant progress in UDGs \cite{Christoph,Erlebach,Huang,Li3}. However, studies on MinWCDSs in general graphs are in a different situation. In 1999, Guha and Khuller \cite{Guha1} designed a $(1.35+\varepsilon)\ln n$-approximation algorithm for MinWCDSs, where $n$ is the number of nodes. In real applications, $\delta_{\max}$ might be much smaller than $n$ where $\delta_{\max}$ is the maximum degree of the input graph. Therefore, one usually expects to replace
$\ln n$ by $\ln \delta_{\max}$.
However, this expectation became a long-standing open problem. In fact, techniques provided in \cite{Guha1}
do not have enough power to do so. Until 2018, with discovery of different techniques, Zhou et al. \cite{Zhou3} presented an $(H(\delta_{\max}+m)+2H(\delta_{\max}-1))$-approximation algorithm for the minimum-weight $(1,m)$CDS (MinW$(1,m)$CDS) problem, where $H(\cdot)$ is the Harmonic function, i.e., $H(k)=\sum_{i=1}^k \frac{1}{i}\leq \ln k+1$ (however, there is a flaw in this work, please see discussion in Section \ref{sec.5})

In this paper, using a completely new idea of analysis, we design a new algorithm for the MinW$(1,m)$CDS problem in a general graph to achieve approximation ratio $2H(\delta_{\max}+m-1)$. Note that our ratio is better than that in \cite{Zhou3} even if its flaw can be fixed.

\section{Preliminaries}\label{secPreliminaries}

We first give a formal definition of the problem and  some preliminary results.

Let $G=(V,E)$ be a connected graph and $C$ be a node subset of $V$. Denote by $G[C]$ the subgraph of $G$ induced by $C$, $N_C (u)$ the set of neighbors of $u$ in $C$, $N(u)=N_V(u)$, and $deg(u)=\mid N(u)\mid $. For $C\subseteq V$, $N(C)=(\bigcup_{u\in C}N(u))\setminus C$ denotes the open neighborhood of $C$.

The formal definition of the MinW$(k,m)$CDS problem is as follows.

\begin{definition}[the minimum weight $k$-connected $m$-fold dominating set (MinW$(k,m)$CDS) problem]
{\rm Let $G$ be a connected graph on node set $V$ and edge set $E$, $k$ and $m$ be two positive integers, and $c:V\rightarrow R^{+}$ be a cost function on the nodes. A node subset $C\subseteq V$ is a $(k,m)$CDS if every node in $V\setminus C$ is adjacent to at least $m$ nodes of $C$, and $G[C]$, the subgraph of $G$ induced by $C$, is $k$-connected (that is, $G[C]$ remains connected after removing at most $k-1$ nodes). The MinW$(k,m)$CDS problem aims to find a $(k,m)$CDS with the minimum cost, where the cost of node set $C$ is $c(C)=\sum_{v\in C}c(v)$.}
\end{definition}

A set function $f:2^V\rightarrow \mathbb R^+$ is {\em monotone nondecreasing} if $f(A)\leq f(B)$ for any $A\subseteq B\subseteq V$; it is {\em submodular} if $f(A\cup B)+f(A\cap B)\leq f(A)+f(B)$ for any $A,B\subseteq V$. For node sets $A,B\subseteq V$, let
$$
\Delta_{A}f(B)=f(A\cup B)-f(B)
$$
be the {\em marginal profit of $A$ over $B$}. The following results are well-known properties for monotone and submodular functions (see, for example, \cite{Du1}).

\begin{lemma}\label{lem0413-1}
A set function $f$ is monotone nondecreasing if and only if $\Delta_{u}f(A)\geq 0$ holds for any $A\subseteq V$ and $u\in V$; it is submodular if and only if $\Delta_uf(A)\geq\Delta_{u}f(B)$ holds for any $A\subseteq B\subseteq E$ and $u\in V\setminus B$; it is monotone nondecreasing and submodular if and only if $\Delta_uf(A)\geq\Delta_{u}f(B)$ holds for any $A\subseteq B\subseteq E$ and $u\in V$.
\end{lemma}

The following is a property of a submodular function.

\begin{lemma}\label{lem0414-1}
If $f:2^V\mapsto\mathbb R^+$ is a submodular function, then for any subsets $A,B\subseteq V$,
$$
\Delta_Bf(A)\leq\sum_{v\in B\setminus A}\Delta_vf(A).
$$
\end{lemma}

\section{Main Results}\label{secAA}

Let us first describe our algorithm and then give analysis.

\subsection{Algorithm}\label{sec.3}

The algorithm uses a greedy strategy. In each iteration, it selects a most cost-effective star. The cost-effectiveness of a star depends on a potential function $g$ designed as follows.

For a node subset $C\subseteq V$ and a node $u\in V$, define
\begin{align*}
& q_{C}(u)=\left \{
\begin{array}{llll}
\max\{0,m-\mid N_{C}(u)\mid \},&&&\mbox   {$u\in V\setminus C$,}\\
0,&&&\mbox{$u\in C$. }
\end{array}
\right.\\
& q(C)=\sum_{u\in V\setminus C}q_{C}(u),\\
& p(C)=\mbox{the number of components of $G[C]$},\\
& f(C)=p(C)+q(C).
\end{align*}
For a node set $U\subseteq V\setminus C$, denote by $NC_{C}(U)$ the
set of components of $G[C]$ which are adjacent to $U$.
Every component in $NC_{C}(U)$ is called a {\em component neighbor}
of $U$ in $C$ (if a component of $G[C]$ has nonempty intersection with $U$, then it is also viewed as a component neighbor of $U$). For a node $u\in V$, we use $S_u$ to denote some star
with center $u$, that is, $S_u$ is a subgraph of $G$ induced by some
edges between node $u$ and some of $u$'s neighbors. In particular, a
single node is a {\em trivial star}. In the following, we treat
$S_u$ as a star as well as the set of nodes in the star. For a node
set $C$, a node $u\in V\setminus C$, and a star $S_u$, suppose
$S_u\setminus\{u\}=\{u_{1},\ldots,u_{s}\}$ has $c(u_1)\leq
\cdot\cdot\cdot\leq c(u_s)$, define
\begin{equation}\label{eq3-15-3}
b_{C}^{S_{u}}(u_{i})=\left \{
\begin{array}{ll}
0,&\mbox   {$q_C(u_i)>0$,}\\
\min\{1,-\Delta_{u_{i}}f(C_i)\},&\mbox{$q_C(u_i)=0$ }
\end{array}
\right.
\end{equation}
where $C_i=C\cup\{u,u_{1},\ldots,u_{i-1}\}$. Let
\begin{equation}\label{eq11-24-2}
g_{C}(S_u)=
-\Delta_{u}f(C)+\sum_{i=1}^{s}b_{C}^{S_{u}}(u_{i}).
\end{equation}
The {\em cost-effectiveness} of star $S_u$ with respect to a node set $C$ is defined to be $g_{C}(S_{u})/c(S_{u})$.

Pseudo codes of the main algorithm is presented in Algorithm \ref{alg1}. It iteratively adds a {\em most cost-effective} star to the current set $C$. We shall show latter in Lemma \ref{le11-23-1} that such a star can be found efficiently by Algorithm \ref{alg2}.

\begin{algorithm}[h]
\caption{\textbf{}}
 Input: A connected graph $G=(V,E)$.

 Output: A node set $C$ which is a $(1,m)$-CDS of $G$.
 \begin{algorithmic}[1]
    \State Set $C\leftarrow\emptyset$.
    \While{$\exists$ a star $S_u$ with $g_C(S_u)>0$}
       \State Use Algorithm \ref{alg2} to compute a most cost-effective star $S_{u}=\arg\max\limits_{S_u\subseteq V\setminus C}\frac{g_C(S_u)}{c(S_u)}$.
       \State $C \leftarrow C\cup S_{u}$
    \EndWhile
    \State Output  $C$.
\end{algorithmic}\label{alg1}
\end{algorithm}

\begin{algorithm}[h]
\caption{\textbf{}}
Input: A connected graph $G=(V,E)$, a node set $C\subseteq V$.

Output: A most cost-effective star $S_u$ with respect to $C$.
 \begin{algorithmic}[1]
    \For {each $u\in V\setminus C$}
        \State $S_{u}\leftarrow\{u\}$
        \If{ $q_{C}(u)=0$}
            \State $N_{u}\leftarrow$ the set of nodes in $N(u)$ satisfying $(\romannumeral3)$ and $(\romannumeral4)$ of Lemma \ref{le11-23-1}.
            \State Order the nodes in $N_{u}$ as $u_{1},\ldots,u_{s}$ such that $c(u_{1})\leq \cdots\leq c(u_{s})$.
            \For{$j=1,\ldots,s$}
                \State If $b_C^{S_u}(u_j)=1$ and $\frac{1}{c(u_j)}
\geq\frac{g_{C}(S_{u})}{c(S_{u})}$, then $S_u\leftarrow S_u\cup \{u_{j}\}$. \label{line0330-4}
%                \State if the unique component in $NC_C(u_j)$ is not adjacent with any node of $\{u,u_1,\ldots,u_{j-1}\}$ and $\frac{1}{c(u_j)}
%\geq\frac{g_{C}(S_{u})}{c(S_{u})}$, then $S_u\leftarrow S_u\cup \{u_{j}\}$ \label{line0330-4}

            \EndFor
        \EndIf
    \EndFor
    \State Output $S_u\leftarrow \arg\max\{g_C(S_u)/c(S_u):u\in V\setminus C\}$, giving priority to trivial star.
\end{algorithmic}\label{alg2}
\end{algorithm}

\subsection{Finding a Most Cost-Effective Star}

Before showing how to find out a most cost-effective star, we first give some properties for functions $p,q$ and $f$.

\begin{lemma}\label{lem0330-1}
Set functions $-q(C)$, $-p(C)$ and $-f(C)$ satisfy the following properties.

$(a)$ $-q(C)$ is monotone nondecreasing and submodular;

$(b_1)$ for any node set $C$ and node $u\not\in C$, $-\Delta_up(C)\geq -1$, equality holds if and only if $u$ is not adjacent with $C$;

$(b_2)$ for any connected node set $C'$ and any node $u\in V\setminus (C\cup C')$, we have
$$
-\Delta_up(C\cup C')\leq -\Delta_up(C)+1,
$$
equality holds only when $G[C']$ is not adjacent with $G[C]$ and node $u$ is adjacent with $G[C']$;

$(c)$ $-f(C)$ is monotone nondecreasing.
\end{lemma}
\begin{proof}
For any node sets $C_{1}\subseteq C_{2}\subseteq V$ and node $u\in V$, we have $q_{C_{1}}(u)\geq q_{C_{2}}(u)$. So, $q(C_{1})=\sum_{u\in V\backslash C_{1}}q_{C_{1}}(u)\geq \sum_{u\in V\backslash C_{2}}q_{C_{2}}(u)=q(C_{2})$ and thus $-q$ is monotone nondecreasing. Furthermore, for any $u\in V\setminus C_2$, we have $-\Delta_{u}q(C_{1})=q_{C_{1}}(u)+\mid \{v\in N(u):q_{C_{1}}(v)>0\}\mid \geq q_{C_{2}}(u)+\mid \{v\in N(u):q_{C_{2}}(v)>0\}\mid =-\Delta_{u}q(C_{2})$, and thus $-q$ is submodular. Property $(a)$ is proved.

Note that adding node $u$ into $C$ will merge those components of $G[C]$ which are adjacent with $u$ into one big component of $G[C\cup\{u\}]$. So,
\begin{equation}\label{eq0407-1}
-\Delta_up(C)=\mid NC_C(u)\mid -1\geq -1.
\end{equation}
Inequality becomes equality if and only if $\mid NC_C(u)\mid =0$, which holds if and only if $u$ is not adjacent with $C$. Property $(b_{1})$ is proved.

By equality \eqref{eq0407-1}, $-\Delta_up(C\cup C')-(-\Delta_up(C))=\mid NC_{C\cup C'}(u)\mid -\mid NC_C(u)\mid $. Since $C'$ is a connected node set, adding $C'$ into $C$ will merge those components of $G[C]$ which are adjacent with $C'$ or have nonempty intersection with $C'$ into one component. So, $\mid NC_{C\cup C'}(u)\mid >\mid NC_C(u)\mid $ happens only when $G[C']$ is a component of $G[C\cup C']$ and $u$ is adjacent with $C'$, in which case $\mid NC_{C\cup C'}(u)\mid =\mid NC_C(u)\mid +1$ and  $-\Delta_up(C\cup C')-(-\Delta_up(C))=1$. In all the other cases, we have $-\Delta_up(C\cup C')-(-\Delta_up(C))\leq 0$. Hence property $(b_2)$ is proved.

By properties $(a),(b_1)$ and Lemma \ref{lem0413-1}, we have $-\Delta_uf(C)=-\Delta_uq(C)-\Delta_up(C)\geq -1$, and equality holds only when $-\Delta_uq(C)=0$ and $-\Delta_up(C)=-1$. Since $-\Delta_up(C)=-1$ implies that $u$ is not adjacent with $C$ and thus $q_C(u)=m$, we have $-\Delta_uq(C)\geq q_C(u)>0$. So, $-\Delta_uf(C)\geq 0$ holds for any node set $C$ and any $u\in V\setminus C$, which is equivalent to say that $-f$ is monotone nondecreasing.
\end{proof}

As a corollary of the above lemma, we have the following result.
\begin{lemma}\label{eq0402-6}
Let $C$ be a node set and $S_u$ be a star rooted at $u$. Suppose $S_u\setminus\{u\}=\{u_1,\ldots,u_s\}$ and $c(u_1)\leq c(u_2)\leq \cdots \leq c(u_s)$. For any $i\in\{1,\ldots,s\}$, denote by $prec(u_i)=\{u,u_1,\ldots,u_{i-1}\}$. Then $-\Delta_{u_{i}}f(C\cup prec(u_i))\geq 0$ and equality holds if and only if $-\Delta_{u_{i}}q(C\cup prec(u_i))=-\Delta_{u_{i}}p(C\cup prec(u_i))=0$.
\end{lemma}
\begin{proof}
By the monotonicity of $-q$, we have $-\Delta_{u_{i}}q(C\cup prec(u_i))\geq0$. Since $prec(u_i)$ is a connected set and $u_{i}$ is adjacent with $prec(u_i)$, by property $(b_1)$ of Lemma \ref{lem0330-1}, we have $-\Delta_{u_{i}}p(C\cup prec(u_i))\geq0$. So, $-\Delta_{u_{i}}f(C\cup prec(u_i))\geq 0$, and $-\Delta_{u_{i}}f(C\cup prec(u_i))=0$ if and only if $-\Delta_{u_{i}}q(C\cup prec(u_i))=-\Delta_{u_{i}}p(C\cup prec(u_i))=0$.
\end{proof}

A simple relation will be used in the proof: for four positive real numbers $a,b,c,d$:
\begin{equation}\label{eq11-23-1}
\frac{a+b}{c+d}\geq \frac{b}{d}\Longrightarrow \frac{a}{c}\geq\frac{a+b}{c+d}\geq\frac{b}{d}.\end{equation}
This is because $\frac{a+b}{c+d}=\frac{b(\frac{a}{b}+1)}{d(\frac{c}{d}+1)}\geq \frac{b}{d}$ implies $\frac{a}{b}\geq \frac{c}{d}$, and then implies $\frac{a+b}{c+d}=\frac{a(1+\frac{b}{a})}{c(1+\frac{d}{c})}\leq \frac{a}{c}$.

The next lemma shows that there exists a most cost-effective star which has some special properties.

\begin{lemma}\label{le11-23-1}
Let $C$ be a node set of graph $G$. There exists a most cost-effective star $S_{u}$ with respect to $C$ such that for any $v\in V(S_u)\setminus \{u\}$, the following properties hold:

$(\romannumeral1)$ $b_{C}^{S_{u}}(v)=1$;

$(\romannumeral2)$ $\frac{1}{c(v)}\geq g_{C}(S_u)/c(S_{u})$;

$(\romannumeral3)$ $q_C(v)=0$;

%$(\romannumeral4)$ {\color{blue} $-\Delta_{v}f(C)=-\Delta_vq(C)=-\Delta_vp(C)=0$;}

$(\romannumeral4)$ $\mid NC_{C}(v)\mid =1$ and the component of $G[C]$ adjacent with $v$ is not adjacent with $u$.
\end{lemma}
\begin{proof}
Let $S_u$ be a most cost-effective star. If $V(S_u)=\{u\}$, then $S_u$ satisfies the above conditions.
In the following, we assume that there is no trivial most cost-effective star. Suppose  $V(S_u)=\{u,u_{1},u_{2},\ldots,u_{s}\}$ and $c(u_{1})\leq c(u_{2})\leq\ldots\leq c(u_{s})$.

{\em Proof of property $(\romannumeral1)$.} We first show that for any $1\leq i\leq s$,
\begin{equation}\label{eq0316-1}
b_{C}^{S_{u}}(u_{j})\leq b_{C}^{S_u-u_i}(u_{j}).
\end{equation}
If $q_C(u_j)>0$, then both $b_{C}^{S_{u}}(u_{j})=b_{C}^{S_u-u_i}(u_{j})=0$, and \eqref{eq0316-1} trivially holds. So, suppose $q_C(u_j)=0$. In this case, $b_C^{S_u}(u_j)$ is determined by $-\Delta_{u_{j}}f(C\cup\{u,u_{1},\ldots,u_{j-1}\})$. By Lemma \ref{lem0330-1} and Lemma \ref{lem0413-1}, for any $1\leq i\leq s$,
\begin{align}
& -\Delta_{u_{j}}q(C\cup\{u,u_{1},\ldots,u_{j-1}\}) \leq -\Delta_{u_{j}}q(C\cup\{u,u_{1},\ldots,u_{i-1},u_{i+1},\ldots,u_j\})\ \mbox{and} \label{eq0613-1} \\
& -\Delta_{u_{j}}p(C\cup\{u,u_{1},\ldots,u_{j-1}\})
 \leq -\Delta_{u_{j}}p(C\cup\{u,u_{1},\ldots,u_{i-1},u_{i+1},\ldots,u_j\})+1.\label{eq0330-2}
\end{align}
In fact, \eqref{eq0330-2} can be improved to
\begin{align*}
-\Delta_{u_{j}}p(C\cup\{u,u_{1},\ldots,u_{j-1}\})
\leq -\Delta_{u_{j}}p(C\cup\{u,u_{1},\ldots,u_{i-1},u_{i+1},\ldots,u_j\}),
\end{align*}
because $u_j$ is adjacent with $u$ and thus the inequality in $(b_1)$ of Lemma \ref{lem0330-1} is strict. Combining this with \eqref{eq0613-1}, we have
\begin{align}\label{eq0330-3}
-\Delta_{u_{j}}f(C\cup\{u,u_{1},\ldots,u_{j-1}\})
\leq -\Delta_{u_{j}}f(C\cup\{u,u_{1},\ldots,u_{i-1},u_{i+1},\ldots,u_j\}),
\end{align}
and thus \eqref{eq0316-1} is proved.

Next, we show that
\begin{equation}\label{eq11-23-2}
\mbox{ for any $1\leq i\leq s$, } b_{C}^{S_{u}}(u_{i})\neq 0.
\end{equation}
Suppose \eqref{eq11-23-2} is not true, and $j$ is an index with $b_{C}^{S_{u}}(u_{j})=0$. Then by \eqref{eq0316-1},
\begin{align}\label{eq11-23-6}
g_{C}(S_{u})
&=-\Delta_{u}f(C)+\sum_{i=1}^{s}b_{C}^{S_{u}}(u_{i})\notag\\
&\leq-\Delta_{u}f(C)+\sum_{i=1}^{j-1}b_{C}^{S_u-u_j}(u_{i})+\sum_{i=j+1}^{s}b_{C}^{S_u-u_j}(u_{i})\notag\\
&=g_{C}(S_{u}-u_{j}).
\end{align}
Combining this with $c(S_u)=c(S_u-u_j)+c(u_j)>c(S_u-u_j)$, we have
\begin{equation}\label{eq2-15-3}
\frac{g_{C}(S_u-u_j)}{c(S_u-u_j)}> \frac{g_{C}(S_{u})}{c(S_{u})},
\end{equation}
and thus $S_u-u_j$ is a more cost-effective star than $S_u$, contradicting the assumption on $S_{u}$. So property \eqref{eq11-23-2} is proved.

By the definition of $b_{C}^{S_{u}}(u_{j})$ in (\ref{eq3-15-3}), we have $b_{C}^{S_{u}}(u_{j})\in\{0,1\}$, and thus property $(\romannumeral1)$ follows from \eqref{eq11-23-2}.

{\em Proof of property $(\romannumeral2)$.} We prove that for any $j=1,\ldots,s$,
\begin{equation}\label{eq11-23-7}
\frac{1}{c(u_{j})}\geq\frac{g_{C}(S_{u})}{c(S_{u})}.
\end{equation}
Note that for any $i>j$, we have $b_C^{S_u}(u_j)=b_C^{S_u-u_i}(u_j)$. For any $i<j$, by \eqref{eq0316-1} and property $(\romannumeral1)$, we have $1=b_{C}^{S_{u}}(u_{j})\leq b_{C}^{S_u-u_i}(u_{j})\leq 1$, and thus $b_{C}^{S_{u}}(u_{j})=b_{C}^{S_u-u_i}(u_{j})=1$. In other words, $b_{C}^{S_{u}}(u_{j})=b_{C}^{S_u-u_i}(u_{j})$ for any $i\neq j$. Hence
$$
g_C(S_u)=g_C(S_u-u_j)+b_C^{S_u}(u_j)=g_C(S_u-u_j)+1.
$$
Since $S_{u}$ is a most cost-effective star, we have
$$
\frac{g_{C}(S_u-u_j)}{c(S_u-u_j)}\leq \frac{g_{C}(S_{u})}{c(S_{u})}
=\frac{g_{C}(S_u-u_j)+1}{c(S_{u}-u_j)+c(u_j)},
$$
Combining this with (\ref{eq11-23-1}), we have inequality \eqref{eq11-23-7}. Thus property $(\romannumeral2)$ is proved.

{\em Proof of property $(\romannumeral3)$.} This property directly follows from property $(\romannumeral1)$ and the definition of $b_C^{S_u}$ in \eqref{eq3-15-3}.

Before proving property $(\romannumeral4)$, we first prove
\begin{equation}\label{eq0401-1}
-\Delta_{u_j}f(C)=-\Delta_{u_j}q(C)=-\Delta_{u_j}p(C)=0.
\end{equation}
Suppose $j$ is an index with $-\Delta_{u_j}f(C)\neq 0$. Then by the monotonicity of $-f$ (see Lemma \ref{lem0330-1}), we have $-\Delta_{u_j}f(C)\geq 1$. Combining this with inequality (\ref{eq11-23-7}), for the trivial star $\{u_j\}$, we have
$$
\frac{g_C(\{u_j\})}{c(u_j)}=\frac{-\Delta_{u_j}f(C)}{c(u_j)}\geq \frac{1}{c(u_{j})}\geq\frac{g_{C}(S_{u})}{c(S_{u})},
$$
which implies that $\{u_j\}$ is a most cost-effective star, contradicting our assumption that there is no trivial most cost-effective star. So, $-\Delta_{u_j}f(C)=0$.

Notice that $q_C(u_j)=0$ means $u_j$ is dominated by at least $m$ nodes in $C$. So, $u_j$ is adjacent with $C$, and thus by $(b_1)$ of Lemma \ref{lem0330-1}, we have $-\Delta_{u_j}p(C)\geq 0$. By the monotonicity of $-q$, we have $-\Delta_{u_j}q(C)\geq 0$. Hence in order that $-\Delta_{u_j}f(C)=0$, we must have $-\Delta_{u_j}p(C)=0$ and $-\Delta_{u_j}q(C)=0$. Equalities in \eqref{eq0401-1} are proved.

{\em Proof of property $(\romannumeral4)$.}
Notice that $-\Delta_{u_j}p(C)=0$ implies that $u_j$ is adjacent with exactly one component of $G[C]$. If this component is also adjacent with $u$, then we also have $-\Delta_{u_{j}}p(C\cup prec(u_j))=0$ (notice that by the star structure, $prec(u_j)$ are in a same component of $G[C\cup prec(u_j)]$). By the monotonicity and submodularity of $-q$, we have $0\leq -\Delta_{u_{j}}q(C\cup prec(u_j))\leq -\Delta_{u_j}q(C)=0$, and thus $-\Delta_{u_{j}}q(C\cup prec(u_j))=0$. But then $-\Delta_{u_{j}}f(C\cup prec(u_j))=0$ and thus $b_C^{S_u}(u_j)=0$, contradicting property $(\romannumeral1)$. Hence the unique component of $G[C]$ adjacent with $u_j$ is not adjacent with $u$. The proof is completed.
\end{proof}

The following lemma shows that a most cost-effective star can be found efficiently.

\begin{lemma}\label{le09-01-1}
For a node set $C$ of graph $G$, a most cost-effective star satisfying Lemma \ref{le11-23-1} can be found in time $O(n^{2})$, where $n$ is the number of nodes in $G$.
\end{lemma}
\begin{proof}
The computation method is described in Algorithm \ref{alg2}. For each $u\in V\setminus C$, the algorithm finds a {\em most cost-effective star centered at $u$}, which satisfies Lemma \ref{le11-23-1} (this will be proved in the following), denote it as $S_u$. A most cost-effective star with respect to $C$ is the best one of $\{S_u\colon u\in V\setminus C\}$. The reason why priority is given to trivial star is: if the output is a nontrivial star, then no trivial star is most cost-effective, and property \eqref{eq0401-1} holds, which brings more structural property to be used in the analysis.

The algorithm is illustrated by the example in Fig. \ref{fig0322-2}, and the proof of the correctness is divided into two steps.

\begin{figure}[!htbp]
\begin{center}
\begin{picture}(110,70)
%\put(0,0){\dashbox(110,70)}
\put(10,10){\circle{20}}\put(70,10){\circle{20}}\put(100,10){\circle{20}}
\put(10,35){\circle*{5}}\put(40,35){\circle*{5}}\put(70,35){\circle*{5}}\put(100,35){\circle*{5}}
\put(55,60){\circle*{5}}
\put(10,60){\circle{20}}\put(100,60){\circle{20}}
{\linethickness{0.45mm}
\qbezier(55,60)(32,47)(10,35)
\qbezier(55,60)(62,47)(70,35)
}
\qbezier(10,20.5)(10,27)(10,35)\qbezier(70,20.5)(70,27)(70,35)
\qbezier(17,17)(28,26)(40,35)\qbezier(100,20)(100,27)(100,35)
\qbezier(55,60)(47,47)(40,35)\qbezier(55,60)(77,47)(100,35)
\qbezier(55,60)(37,60)(20,60)\qbezier(55,60)(72,60)(90,60)

\put(52,64){$u$}\put(13,31){$u_1$}\put(43,33){$u_2$}\put(73,33){$u_3$}\put(103,33){$u_4$}
\end{picture}
\caption{An illustration of the execution of Algorithm \ref{alg2}. Every $u_i$ ($i=1,2,3,4$) is adjacent with exactly one component of $G[C]$ (indicated by big circle) which is not adjacent with the center $u$. Suppose $c(u_1)\leq c(u_2)\leq c(u_3)\leq c(u_4)$ and only $u_4$ has its cost $c(u_4)>\frac{c(S_{u}^{curr})}{g_{C}(S_{u}^{curr})}$. The blackened structure is the final $S_u$. The reason why node $u_2$ is not added into $S_u$ is because $b_C^{S_u}(u_2)=0$. Node $u_4$ is not added into $S_u$ because its cost is too large to satisfy property $(\romannumeral2)$.}
\label{fig0322-2}
\end{center}
\end{figure}
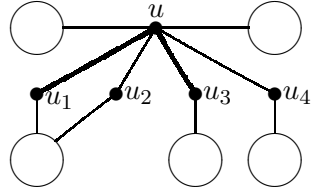

{\bf Claim 1.} The star $S_u$ computed by the algorithm satisfies the four properties described in Lemma \ref{le11-23-1}.

Notice that if $q_{C}(u)>0$, then for any neighbor $v$ of $u$, we have $-\Delta_{v}q(C)>0$, violating property $(\romannumeral3)$. Hence, only when $q_C(u)=0$, we need to consider a nontrivial star (through line 3 to line 9 of Algorithm \ref{alg2}).

Denote by $N_{u}$ the set of nodes in $N(u)$ satisfying properties $(\romannumeral3)$ and $(\romannumeral4)$. 
The feet of $S_u$ can only be taken from $N_u$. The idea of the algorithm is to start from the trivial star $S_u=\{u\}$, and sequentially check nodes of $N_u$ in increasing order of costs. If properties $(\romannumeral1)$ and $(\romannumeral2)$ are satisfied, then expand $S_u$. What needs to be explained is:
%why the conditions in line \ref{line0330-4} of Algorithm \ref{alg2} guarantees properties $(\romannumeral1)$ and $(\romannumeral2)$.
%
%In fact, condition $\mid NC_C(u_j)\mid =1$ guarantees $-\Delta_{u_j}p(C)=0$. Combining this with $-\Delta_{u_j}q(C)=0$, we have $-\Delta_{u_j}f(C)=0$.
%
%The first half condition in line \ref{line0330-4} of the algorithm guarantees $-\Delta_{u_j}p(C\cup \{u,u_1,\ldots,u_{j-1}\})=1$ (that is, after adding $u_j$, the unique component in $NC_C(u_j)$ is merged into the component of $G[C\cup \{u,u_1,\ldots,u_j\}]$ containing $u$). By the monotonicity and submodularity of $q$ and the condition $-\Delta_{u_j}q(C)=0$, we have $-\Delta_{u_j}q(C\cup \{u,u_1,\ldots,u_{j-1}\})=0$. So, $-\Delta_{u_j}f(C\cup \{u,u_1,\ldots,u_{j-1}\})=1$, and thus $b_C^{S_u}(u_j)=1$. It should be remarked that by reversing the above arguments,
%\begin{equation}\label{eq0331-1}
%\mbox{properties $(\romannumeral1)$, $(\romannumeral4)$ and $(\romannumeral5)$ lead to the first half condition of line \ref{line0330-4}.}
%\end{equation}}
why $\frac{1}{c(u_j)}\geq
\frac{g_{C}(S_{u}^{curr})}{c(S_{u}^{curr})}$ for the {\em current
star} $S_u^{curr}$ implies $\frac{1}{c(u_j)}\geq
\frac{g_{C}(S_{u}^{final})}{c(S_{u}^{final})}$ for the final star
$S_u^{final}$ computed by the algorithm. Suppose $u_\ell$ is the
last node added into $S_{u}$. Denote
$S_u^{u_\ell}=S_u^{final}-u_\ell$. Since $u_\ell$ is eligible to be
added by the algorithm, we have $\frac{1}{c(u_\ell)}\geq
\frac{g_{C}(S_{u}^{u_\ell})}{c(S_{u}^{u_\ell})}$ and
$b_C^{S_u}(u_\ell)=1$. It follows that
$g_C(S_u^{final})=g_C(S_u^{u_\ell})+1$, and thus
$$
\frac{g_C(S_u^{final})}{c(S_u^{final})}=\frac{g_C(S_u^{u_\ell})+1}{c(S_u^{u_\ell})+c(u_\ell)}.
$$
Then by (\ref{eq11-23-1}), we have
\begin{equation}\label{eq0401-4}
\frac{1}{c(u_\ell)}\geq \frac{g_{C}(S_{u}^{final})}{c(S_{u}^{final})}\geq \frac{g_{C}(S_{u}^{u_\ell})}{c(S_{u}^{u_\ell})}.
\end{equation}
Since $c(u_j)\leq c(u_\ell)$, we also have $\frac{1}{c(u_j)}\geq\frac{g_{C}(S_{u}^{final})}{c(S_{u}^{final})}$.

It should be remarked that similarly to the derivation for the right side of inequality \eqref{eq0401-4}, by induction on the feet of $S_u$ in the reverse order of their addition into $S_u$, it can be seen that $g_C(S_u^i)/c(S_u^i)\leq g_C(S_u^j)/c(S_u^j)$ for $i<j$, where $S_u^i$ is the current star when $u_i$ is added. As a corollary,
\begin{equation}\label{eq0401-5}
\frac{g_C(S_u^{final})}{c(S_u^{final})}\geq \frac{g_C(S_u^{curr})}{c(S_u^{curr})}
\end{equation}
throughout the process.

{\bf Claim 2.} The computed star $S_u$ is indeed most cost-effective.

Let $S_{u}^*$ be a most cost-effective star centered at $u$ which satisfies those properties in Lemma \ref{le11-23-1}. We shall prove that
\begin{equation}\label{eq1-21-3}
\frac{g_{C}(S_{u})}{c(S_{u})}= \frac{g_{C}(S_u^*)}{c(S_u^*)}.
\end{equation}

First consider the case that $S_u^*\subseteq S_{u}$. If $S_u^*=S_u$, then \eqref{eq1-21-3} is obviously true. So, suppose $S_{u}\setminus S_u^*\neq \emptyset$. Let $u_j$ be a maximum-cost node of $S_{u}\setminus S_u^*$. By property $(\romannumeral2)$, $\frac{1}{c(u_j)}\geq\frac{g_{C}(S_{u})}{c(S_{u})}$. Notice that any star $S_v$ satisfying property $(\romannumeral1)$ has $g_C(S_v)=-\Delta_vf(C)+\mid S_v\setminus \{v\}\mid$. So, $g_{C}(S_{u})=g_{C}(S_u^*)+\mid S_{u}\setminus S_u^*\mid $. Then by the assumption that $u_j$ has the maximum cost in $S_u\setminus S_u^*$, we have
$$\frac{\mid S_{u}\setminus S_u^*\mid}{c(S_{u}\setminus S_u^*)}\geq\frac{1}{c(u_j)}\geq\frac{g_{C}(S_{u})}{c(S_{u})}
=\frac{g_{C}(S_u^*)+\mid S_{u}\setminus S_u^*\mid}{c(S_u^*)+c(S_{u}\setminus S_u^*)}.$$
Then by (\ref{eq11-23-1}), we have $\frac{g_{C}(S_{u})}{c(S_{u})}\geq \frac{g_{C}(S_u^*)}{c(S_u^*)}$, and thus $\frac{g_{C}(S_{u})}{c(S_{u})}=\frac{g_{C}(S_u^*)}{c(S_u^*)}$ by the optimality of $S_u^*$.

Next, consider the case when $S_u^*\setminus S_{u}\neq \emptyset$. Consider a node $u_j\in S_u^*\setminus S_{u}$. By Lemma \ref{le11-23-1} and observation \eqref{eq0401-5},
$$
\frac{1}{c(u_j)}\geq \frac{g_{C}(S_u^*)}{c(S_u^*)}\geq \frac{g_{C}(S_{u})}{c(S_{u})}\geq \frac{g_{C}(S_{u}^{curr})}{c(S_{u}^{curr})}.
$$
So, the reason why $u_j$ is not added into $S_{u}$ is because $b_C^{S_u}(u_j)=0$, that is, $-\Delta_{u_j}f(C\cup prev(u_j))=0$. By Lemma \ref{eq0402-6}, we have $-\Delta_{u_j}p(C\cup prev(u_j))=0$, which implies that the unique component of $G[C]$ adjacent with $u_j$ is also adjacent with a node $u_\ell$ with $\ell<j$ which has been added into $S_u$ before. Note that for those nodes in $N_u$ which are adjacent with the same component of $G[C]$, in order that $b_C^{S_u^*}$ has value 1, at most one of them can belong to $S_u^*$. So, $u_\ell\not\in S_u^*$. Let $S'_u=S_u^*+u_\ell-u_j$. Then $g_{C}(S_u^*)=g_{C}(S'_u)$. Since $u_\ell$ is ordered before $u_j$, we have $c(u_\ell)\leq c(u_j)$. Then
$$\frac{g_{C}(S_u^*)}{c(S_u^*)}=\frac{g_C(S'_u)}{c(S'_u)+c(u_j)-c(u_\ell)}\leq \frac{g_C(S'_u)}{c(S'_u)}\leq \frac{g_{C}(S_u^*)}{c(S_u^*)}.$$
It follows that $c(u_\ell)=c(u_j)$ and $\frac{g_{C}(S_u^*)}{c(S_u^*)}=\frac{g_C(S'_u)}{c(S'_u)}$, which implies that $S_u'$ is also a most cost-effective star satisfying Lemma \ref{le11-23-1}. Notice that $S_u'$ and $S_u$ have one more common foot. Proceeding like this, by an inductive argument on $\mid S_u^*\setminus S_u\mid$, it can be shown that $S_u$ is also most cost-effectiveness.

The correctness of Algorithm \ref{alg2} follows from Claim 1 and Claim 2. As to the time complexity, note that for any $u\in V\setminus C$, the algorithm only considers $u$ and $N(u)$ at most once, so the time spent by the algorithm is at most $2\mid E(G)\mid$, which is $O(n^{2})$.
\end{proof}

\subsection{Feasibility and Approximation Ratio}\label{sec.4}

Before proving that the output of Algorithm \ref{alg1} is a feasible solution with the desired approximation ratio, we first prove two technical lemmas.

\begin{lemma}\label{coro0328-1}
Let $C$ be a node set and $S_u$ be a most cost-effective star satisfying the properties in Lemma \ref{le11-23-1}. Then $g_C(S_u)=-\Delta_{S_u}f(C)$.
\end{lemma}
\begin{proof}
The lemma holds if $S_u$ is a trivial star. So in the following, we consider nontrivial star. In this case, \eqref{eq0401-1} holds for any node $v\in S_u\setminus \{u\}$. In particular, $-\Delta_{v}q(C)=0$. Then by the monotonicity and submodularity of $-q$, we have
\begin{align*}
0\leq -\Delta_{S_u\setminus\{u\}}q(C\cup\{u\})\leq -\sum_{v\in S_u\setminus\{u\}}\Delta_vq(C\cup\{u\})\leq -\sum_{v\in S_u\setminus\{u\}}\Delta_vq(C)=0.
\end{align*}
So, $-\Delta_{S_u\setminus\{u\}}q(C\cup\{u\})=0$, and thus
$-\Delta_{S_u\setminus\{u\}}f(C\cup
\{u\})=-\Delta_{S_u\setminus\{u\}}p(C\cup \{u\})$. By property $(\romannumeral1)$ and $(\romannumeral4)$ (note that
these properties imply that those unique component neighbors of
distinct feet of $S_u$ are distinct), we have
\begin{align*}
-\Delta_{S_u}f(C)
&=-\Delta_{u}f(C)-\Delta_{S_u\setminus \{u\}}f(C\cup \{u\})\notag\\
&=-\Delta_{u}f(C)-\Delta_{S_u\setminus \{u\}}p(C\cup \{u\})\notag\\
&= -\Delta_{u}f(C)+\mid S_u\setminus \{u\}\mid \notag\\
&=-\Delta_{u}f(C)+\sum_{v\in S_u\setminus \{u\}}b_{C}^{S_{u}}(v)\notag\\
&=g_{C}(S_{u}).
\end{align*}
The lemma is proved.
\end{proof}

\begin{lemma}\label{le1-11-3}
For two node sets $C,C'\subseteq V(G)$, suppose there is an edge $uv$ with $u\in C'\setminus C$, $v\in V\setminus (C\cup C')$, and $q_{C}(v)>0$. Then,
$-\Delta_{C'}q(C)+(-\Delta_{v}q(C))\geq-\Delta_{C'\cup \{v\}}q(C)+1.$
\end{lemma}
\begin{proof}
By the assumption $u, v\not\in C$, $q_C(v)>0$, and $v$ is adjacent with $u$, we have $q_{C}(v)=q_{C\cup \{u\}}(v)+1$. Combining this with the submodularity of $-q$, we have
\begin{align*}
-\Delta_{v}q(C) & = q_C(v)+\mid\{x\in N(v)\colon q_C(x)>0\}\mid\\
& \geq q_{C\cup \{u\}}(v)+1+\mid\{x\in N(v)\colon q_{C\cup\{u\}}(x)>0\}\mid\\
& =-\Delta_{v}q(C\cup \{u\})+1\\
& \geq-\Delta_{v}q(C\cup C')+1.
\end{align*}
It follows that
\begin{align*}
-\Delta_{C'}q(C)+(-\Delta_{v}q(C)) \geq -\Delta_{C'}q(C)-\Delta_{v}q(C\cup C')+1
=-\Delta_{C'\cup\{v\}}q(C)+1.
\end{align*}
The lemma is proved.
\end{proof}

The following result is a folklore for dominating set, which can be found, for example,  in Wan et al. \cite{Wan1}.

\begin{lemma}\label{le1-14-3}
Suppose $C$ is a dominating set of $G$ and $G[C]$ is not connected. Then, the two nearest components of $G[C]$ are at most three hops away.
\end{lemma}

The next lemma shows that the algorithm outputs a feasible solution.

\begin{lemma}\label{le11-23-3}
The output $C$ of Algorithm \ref{alg1} is a $(1,m)$-CDS of graph $G$.
\end{lemma}
\begin{proof}
First, we show that $C$ is an $m$-DS of $G$. If not, then there exists a node $u\in V\setminus C$ with $q_C(u)>0$. If $u$ is adjacent with $C$, then $-\Delta_up(C)\geq 0$ and $-\Delta_uq(C)\geq q_C(u)>0$. In this case, $-\Delta_uf(C)>0$. If $u$ is not adjacent with $C$, since $G$ is connected, we may consider such $u$ which is adjacent with a node $v\in (V\setminus C)\cap N(C)$. In this case, $-\Delta_vp(C)\geq 0$ and $-\Delta_vq(C)>0$ (at least the covering requirement of $u$ is reduced by 1), and thus $-\Delta_uf(v)>0$. In any case, there is a node $x$ with $-\Delta_xf(C)>0$ and thus $S_x=x$ is a star with $g_C(S_x)=-\Delta_xf(C)>0$, which implies that Algorithm \ref{alg1} will not terminate. So, at the termination, $C$ is an $m$-DS.

Next, we show that $G[C]$ is connected. If not, then by Lemma \ref{le1-14-3}, there exists one node $u$ (or two adjacent nodes $u,v$) adding which can connect two components of $G[C]$. Such node $u$ (or adjacent nodes $u,v$) can be viewed as a star $S_u$ with $g_C(S_u)>0$. Hence the algorithm will not terminate if $G[C]$ is not connected.
\end{proof}

\begin{theorem}\label{lem11-24-1}
Let $C^*$ be an optimal solution to a MinW$(1,m)$-CDS instance on graph $G$,  and $C$ be the output of Algorithm \ref{alg1}. Then $c(C)\leq 2H(\delta_{\max}+m-1)c(C^*)$, where $H(\gamma)=\sum_{i=1}^{\gamma}1/i$ is the $\gamma$th Harmonic number and $\delta_{\max}$ is the maximum degree of $G$.
\end{theorem}
\begin{proof}
Let $S_{1},S_{2},\ldots,S_{g}$ be the stars chosen by Algorithm \ref{alg1} in the order of their selection into set $C$. For $i=1,2,\ldots,g$, denote $C_i=S_{1}\cup S_{2}\cup\ldots\cup S_{i}$, and let  $C_0=\emptyset$. Furthermore, let $r_{i}=g_{C_{i-1}}(S_{i})$ and $w_{i}=\frac{c(S_{i})}{r_{i}}$. By Lemma \ref{coro0328-1}, we have
\begin{equation}\label{eq0417-1}
r_{i}=g_{C_{i-1}}(S_{i})=-\Delta_{S_{i}}f(C_{i-1}).
\end{equation}

Suppose $\mid C^*\mid=t$ and $T$ is a spanning tree of $G[C^*]$. Order nodes in $C^*$ as $u_1,\ldots,u_t$ such that a parent is ordered before its children, and brothers are ordered in non-decreasing order of costs. For $i=1,2,\ldots,t$, denote $C_i^*=\{u_1,\ldots,u_i\}$, and let $C_0^*=\emptyset$. Furthermore, let $Y_i$ be the sub-star of $T$ rooted at $u_i$. Then, $T$ is divided into the union of stars $T=Y_1\cup Y_2\cup \cdots \cup Y_t$.

For $i\in\{1,\ldots,g\}$ and $j\in\{1,\ldots,t\}$, let $a_{i,j}=g_{C_{i}}(Y_{j})$ and $w_{i,1}=\ldots=w_{i,r_{i}}=w_{i}$.
%Then
%$$
%c(S_i)=w_{i,1}+\cdots +w_{i,r_{i}}.
%$$
For any integer $1\leq \ell\leq a_{0,j}$, denote $b_{j,\ell}=\frac{c(Y_{j})}{\ell}$.

{\bf The idea for the following proof.}
Let $A=\bigcup_{i=1}^gA_{i}$ with $A_{i}=\{w_{i,1},\ldots,w_{i,r_{i}}\}$, and $B=\{b_{1,1},\ldots,b_{1,a_{0,1}},$ $ b_{2,1},\ldots,b_{2,a_{0,2}},\ldots,b_{t,1},\ldots,b_{t,a_{0,t}}\}$.
If
\begin{align}\label{eq0328-2}
\mbox{there is an injective mapping $h:A\rightarrow B$ } \mbox{such that $w\leq h(w)$ for any } w\in A,
\end{align}
then we shall have
\begin{align*}
c(C_{g})&=\sum_{i=1}^{g}c(S_{i})=\sum_{w\in A}w\leq\sum_{w\in A}h(w)\leq \sum_{j=1}^{t}\sum_{\ell=1}^{a_{0,j}}b_{j,\ell}=\sum_{j=1}^{t}H(a_{0,j})c(Y_{j}),
\end{align*}
where the second equality holds because $c(S_i)=w_ir_i=\sum_{w\in A_i}w$, and the second inequality holds because of ``injection''. Note that $a_{0,j}=g_{\emptyset}(Y_j)$, and $b_\emptyset^{Y_j}(v)=0$ holds for any $v\in Y_j\setminus \{u_j\}$ (since $q_{\emptyset}(v)=m>0$). So,
\begin{equation}\label{eq230221}
a_{0,j}=-\Delta_{u_j}f(\emptyset)=m+deg(u_j)-1\leq m+\delta_{\max}-1.
\end{equation}
Combining this with the fact that every node of $C^{\star}$ appears in at most two $Y_j$'s, we have
$$
c(C_{g})\leq H(\delta_{\max}+m-1)\sum_{j=1}^{t}c(Y_{j})\leq 2H(\delta_{\max}+m-1)c(C^{\star}).
$$

{\bf Constructing a mapping $h$ satisfying \eqref{eq0328-2}.} The construction is based on the following two claims.

{\bf Claim 1.} For any $1\leq i\leq g$,
$\sum_{l=i}^{g}r_l\leq \sum_{j=1}^{t}a_{i-1,j}$.

Using \eqref{eq0417-1} and the fact $f(C_g)=1$, the left-hand side can be written as
\begin{align}\label{eq9-4-3}
\sum_{l=i}^{g}r_{l} & =\sum_{l=i}^{g}(-\Delta_{S_{l}}f(C_{l-1}))=\sum_{l=i}^g(f(C_{l-1})-f(C_l))\nonumber\\
& =f(C_{i-1})-f(C_g)=f(C_{i-1})-1  \nonumber\\
& =q(C_{i-1})+p(C_{i-1})-1.
\end{align}
For each $j\in\{2,\ldots,t\}$, denote by $j^{(p)}$ the index for the parent of node $u_j$ in tree $T$ (superscript $(p)$ indicates ``parent''). Then the right-hand side can be written as
\begin{align}\label{eq11-4-3}
\sum_{j=1}^{t}a_{i-1,j} &=\sum_{j=1}^{t}g_{C_{i-1}}(Y_{j})\notag\\
&=\sum_{j=1}^{t}(-\Delta_{u_{j}}q(C_{i-1})+\mid NC_{C_{i-1}}(u_{j})\mid -1+\sum_{u\in V(Y_{j})\setminus\{u_{j}\} }b_{C_{i-1}}^{Y_j}(u)) \notag\\
&=\sum_{j=1}^{t}(-\Delta_{u_{j}}q(C_{i-1}))+\sum_{j=1}^{t}(\mid NC_{C_{i-1}}(u_{j})\mid)  -t+\sum_{j=2}^{t}b_{C_{i-1}}^{Y_{j^{(p)}}}(u_{j})
\end{align}
where the second equality uses expression \eqref{eq0407-1}.

%{\color{blue} For a node $u_j\in C^*\setminus\{u_1\}$, denote by $prec(u_j)$ the set of nodes in $Y_{j^{(p)}}$ consisting of the parent of $u_j$ and those elder brothers of $u_j$.}

Let $X_{1}=\{u_{j}\in C^*\setminus \{u_{1}\}:b_{C_{i-1}}^{Y_{j^{(p)}}}(u_{j})=1\}$, $X_{2}=\{u_{j}\in C^*\setminus \{u_{1}\}:NC_{C_{i-1}}(u_{j})\cap NC_{C_{i-1}}(C_{j-1}^*)\neq \emptyset\}$, $X_{3}=\{u_{j}\in C^{\star}\setminus \{u_{1}\}:q_{C_{i-1}}(u_{j})>0\}$. Observe that
\begin{equation}\label{eq0402-7}
C^*\setminus\{u_1\}\subseteq X_1\cup X_2\cup X_3.
\end{equation}
In fact, for any node $u_{j}\in C^*\setminus\{u_1\}$, if $u_j\not\in X_1\cup X_3$, then $b_{C_{i-1}}^{Y_{j^{(p)}}}(u_{j})=0$ and $q_{C_{i-1}}(u_{j})=0$. Note that $q_{C_{i-1}}(u_{j})=0$ implies that $NC_{C_{i-1}}(u_j)\neq\emptyset$. In order that $b_{C_{i-1}}^{Y_{j^{(p)}}}(u_j)=0$, we have $-\Delta_{u_j}f(C_{i-1}\cup prec(u_j))=0$. Then by Lemma \ref{eq0402-6}, $-\Delta_{u_j}p(C_{i-1}\cup prec(u_j))=0$, which implies that any component in $NC_{C_{i-1}}(u_j)$ is adjacent with a node in $prec(u_j)\subseteq C_{j-1}^*$. Hence $u_j\in X_2$, relation \eqref{eq0402-7} is proved.

As a consequence of \eqref{eq0402-7}, we have
\begin{equation}\label{eq8-11-3}
\mid X_{1}\mid +\mid X_{2}\mid +\mid X_{3}\mid\geq t-1.
\end{equation}
Furthermore, by definition, we have
\begin{equation}\label{eq0402-8}
\sum_{j=2}^{t}b_{C_{i-1}}^{Y_{j^{(p)}}}(u_{j})=\mid X_1\mid.
\end{equation}

By the submodularity of $-q$, for any $u_j\in C^*\setminus\{u_1\}$,
\begin{equation}\label{eq0403-1}
-\Delta_{C_{j-1}^*}q(C_{i-1})
+(-\Delta_{u_{j}}q(C_{i-1}))\geq -\Delta_{C_j^*}q(C_{i-1}).
\end{equation}
Furthermore, for any node $u_{j}\in X_{3}$, by Lemma \ref{le1-11-3} and because $u_j$ is adjacent with $u_{j^{(p)}}\in C_{j-1}^*$, we have
\begin{equation}\label{eq0403-2}
-\Delta_{C_{j-1}^*}q(C_{i-1})
+(-\Delta_{u_{j}}q(C_{i-1}))\geq -\Delta_{C_j^*}q(C_{i-1})+1.
\end{equation}
Inequalities \eqref{eq0403-1} and \eqref{eq0403-2} can be unified as
$$
-\Delta_{C_{j-1}^*}q(C_{i-1})
+(-\Delta_{u_{j}}q(C_{i-1}))\geq -\Delta_{C_j^*}q(C_{i-1})+{\bf 1}_{u_j\in X_3},
$$
where ${\bf 1}_{u_j\in X_3}$ is the indicator of whether $u_j\in X_3$. Then,
\begin{align}\label{eq9-11-3}
 \sum_{j=1}^{t}(-\Delta_{u_{j}}q(C_{i-1}))
 \geq & \sum_{j=1}^t\big(-\Delta_{C_j^*}q(C_{i-1})+\Delta_{C_{j-1}^*}q(C_{i-1})\big)+\mid X_3\mid \notag\\
 = & -\Delta_{C^*_t}q(C_{i-1})+\Delta_{\emptyset}q(C_{i-1})+\mid X_3\mid\notag\\
 = &\ q(C_{i-1})+\mid X_3\mid,
\end{align}
where the last equality uses the fact $q(C_t^*\cup C_{i-1})=0$ and $-\Delta_{\emptyset}q(C_{i-1})=0$.

For the second item of expression (\ref{eq11-4-3}), using the fact $\mid NC_{C_{i-1}}(C^*)\mid=p(C_{i-1})$ (since $C^*$ dominates every node of $C_{i-1}$) and the fact $\mid NC_{C_{i-1}}(u_{j})\cap NC_{C_{i-1}}(C_{j-1}^*)\mid\geq 1$ for any $u_j\in X_2$ (by the definition of $X_2$), we have
\begin{align}\label{eq10-11-3}
 & \sum_{j=1}^{t}\mid NC_{C_{i-1}}(u_{j})\mid \notag\\
=\ &\sum_{j=1}^{t}(\mid NC_{C_{i-1}}(u_{j})\setminus NC_{C_{i-1}}(C_{j-1}^*)\mid
+\mid NC_{C_{i-1}}(u_{j})\cap NC_{C_{i-1}}(C_{j-1}^*)\mid) \notag\\
\geq\ & \mid NC_{C_{i-1}}(C^*)\mid+\mid X_2\mid=p(C_{i-1})+\mid X_{2}\mid .
\end{align}

Combining inequalities \eqref{eq9-4-3}, (\ref{eq8-11-3}), \eqref{eq0402-8}, (\ref{eq9-11-3}) and (\ref{eq10-11-3}), Claim 1 is proved.

{\bf Claim 2.} $a_{i,j}\leq a_{0,j}$ for any $1\leq i\leq g$ and $1\leq j\leq t$.

Let $N_{1}(u_{j})=\{v\in N(u_{j}):q_{C_{i}}(v)>0\}$, $N_{2}(u_{j})=N(u_{j})\cap C_{i}$, and $N_{3}(u_{j})=\{v\in Y_j:b_{C_{i}}^{Y_{j}}(v)=1\}$. Then $-\Delta_{u_j}q(C_i)=q_{C_i}(u_j)+\mid N_1(u_j)\mid$, $\mid NC_{C_i}(u_{j})\mid \leq \mid N_2(u_j)\mid $, and $\sum_{v\in V(Y_{j})\backslash \{u_{j}\}}b_{C_{i}}^{Y_{j}}(v)=\mid N_3(u_j)\mid $. Notice that $N_{1}(u_{j})$, $N_{2}(u_{j})$, and $N_{3}(u_{j})$ are mutually disjoint. In fact, since $N_3(u_j)\subseteq Y_j\subseteq V\setminus C_i$, we have $N_2(u_j)\cap N_3(u_j)=\emptyset$. By the definition of $b_C^{S_u}(v)$ in \eqref{eq3-15-3}, we have $N_1(u_j)\cap N_3(u_j)=\emptyset$. Since any node $v\in C_i$ has $q_{C_i}(v)=0$, so $N_1(u_j)\cap N_2(u_j)=\emptyset$. Hence $\mid N_1(u_j)\mid +\mid N_2(u_j)\mid +\mid N_3(u_j)\mid \leq deg(u_j)$. Then
\begin{align*}
a_{i,j}&= g_{C_{i}}(Y_{j})=-\Delta_{u_{j}}f(C_{i})+\sum_{v\in V(Y_{j})\backslash \{u_{j}\}}b_{C_{i}}^{Y_{j}}(v)\nonumber\\
&=-\Delta_{u_{j}}q(C_{i})+(\mid NC_{C_i}(u_{j})\mid -1)+\sum_{v\in V(Y_{j})\backslash \{u_{j}\}}b_{C_{i}}^{Y_{j}}(v)\nonumber\\
&\leq (q_{C_{i}}(u_{j})+\mid N_{1}(u_{j})\mid )+(\mid N_{2}(u_{j})\mid -1)+\mid N_{3}(u_{j})\mid \nonumber\\
&\leq q_{C_{i}}(u_{j})+deg(u_j)-1\nonumber\\
&\leq m+deg(u_j)-1=a_{0,j},
\end{align*}
where the last equality uses \eqref{eq230221}. Claim 2 is proved.

{\bf Finishing the construction of $h$ satisfying \eqref{eq0328-2}:}
For $i\in\{1,\ldots,g\}$, let $B_{i}=\bigcup_{j=1}^{t}\{b_{j,1},$ $b_{j,2},\ldots,b_{j,a_{i-1,j}}\}$. By Claim 2, every $B_i$ is well defined and $B_i\subseteq B$. By the greedy choice of $S_{i}$, we have $w_{i}=\frac{c(S_{i})}{r_{i}}\leq \frac{c(Y_{j})}{a_{i-1,j}}$ $(\forall 1\leq j\leq t)$. So,
\begin{equation}\label{eq0403-6}
\mbox{$w_{i}\leq b$ holds for any $b\in B_{i}$.}
\end{equation}
Next, we show that there exists an injection $h$ on $A$ such that
\begin{equation}\label{eq0615-1}
h(A_i)\subseteq B_i\setminus \bigcup_{\ell=i+1}^gh(A_\ell)\ \mbox{for any $i=g,g-1,\ldots,1$,}
\end{equation}
This can be proved by induction on $i$ from $g$ down to $1$. First, using Claim 1 for $i=g$, we have $\mid B_g\mid =\sum_{j=1}^ta_{g-1,j}\geq r_g=\mid A_g\mid $. So, an injection from $A_g$ into $B_g$ exists. Suppose we have established an injection $h$ from $\bigcup_{x=i+1}^g A_x$ into $\bigcup_{x=i+1}^gB_x$ with $h(A_x)\subseteq B_x\setminus \bigcup_{\ell=x+1}^gh(A_\ell)$ for any $x\in\{i+1,\ldots,g\}$. By $\mid B_i\setminus \bigcup_{\ell=i+1}^gh(A_\ell)\mid \geq \sum_{j=1}^ta_{i-1,j}-\sum_{\ell=i+1}^g\mid A_\ell\mid =\sum_{j=1}^ta_{i-1,j}-\sum_{\ell=i+1}^gr_\ell\geq r_i=\mid A_i\mid $, an injection from $A_i$ into $B_i\setminus \bigcup_{\ell=i+1}^gf(A_\ell)$ exists. When $i$ reaches $1$, an injection $h$ satisfying \eqref{eq0615-1} is established. Combining \eqref{eq0615-1} with \eqref{eq0403-6}, an injection $h$ satisfying \eqref{eq0328-2} is found, and the theorem is proved.
\end{proof}

\section{Conclusion and Discussion}\label{sec.5}

CDSs were proposed by Das and Bhargharan \cite{Das} and Ephremides et al. \cite{Ephremides} to serve as virtual backbones in WSNs. There exist many results on CDS in the literature.

In unweighted case, the MinCDS problem in a general graph has received a sequence of efforts \cite{Guha,Ruan,Du}.
The best approximation ratio is $(\ln\delta_{\max}+2)$ in \cite{Ruan} or $(1+\varepsilon)\ln(\delta_{\max}-1)$ in \cite{Du},
where $\varepsilon$ is an arbitrary positive real number. The MinCDS in UDG has polynomial-time approximation shcemes (PTASs) \cite{Cheng,Zhang1}. For the fault-tolerant Min$(k,m)$CDS problem in general graphs, asymptotically tight approximations have been obtained for $k=1,2,3$ and $m\geq k$ \cite{Shi1,Zhou2,Zhou1}. For general constants $m\geq k$, a $(2k-1)\ln \delta_{\max}$-approximation algorithm was proposed by Zhang et al. \cite{Zhang}. For the Min$(k,m)$CDS problem in UDGs, constant approximations have been developed \cite{Shang,Shi1,Wang2,Zhou2}. As for the weighted version of fault-tolerant virtual backbones, Shi
et al. \cite{Shi2} and Fukunaga \cite{Fukunaga} independently
presented constant approximation algorithms for the MinW$(k,m)$CDS
problem in UDGs, and Nutov \cite{Nutov2} proposed an $O(k\ln n)$-approximation algorithm for general graphs, where the constant in the big $O$ is at least 10. A question is: can the constant in $O$ be further reduced?

In weighted case, the MinWCDS problem in UDGs has several constant-approximations \cite{Christoph,Huang,Dai,Erlebach,Zou}. However, it is still open whether there exists a PTAS.
More information can be found in \cite{DuWan,WZLD,ZHANG3}.
For general graphs, progress on the MinWCDS problem is slow. In 1999, Guha and Khuller \cite{Guha1} proposed a $(1.35+\varepsilon)\ln n$-approximation algorithm. Until 2018, \cite{Zhou3} presented an
asymptotic $3\ln\delta_{\max}$-approximation algorithm. However, the analysis in \cite{Zhou3}
contains a flaw. Actually, an inequality in their derivation contains a small error term 1. This small error accumulates to an uncontrollable error in the total weight. Existing methods seem unable to correct this flaw. This is the motivation for the current paper.

%This ratio remains unchanged in almost two decades, until 2018 Zhou et al. \cite{Zhou3} studied a fault-tolerant version of the problem, MinW$(1,m)$CDS, obtaining approximation ratio $H(\delta_{\max}+m)+2H(\delta_{\max}-1)$. This ratio reduces the dependence on $n$ to the dependence on $\delta_{\max}$, at the expense of a larger constant before $\ln$. It is a two-phase algorithm, first finding an $m$-DS $D$, and then adding more connectors.

In this paper, we presented a $2H(\delta_{\max}+m-1)$-approximation algorithm for the MinW$(1,m)$-CDS problem in a general graph. Unlike the algorithm in \cite{Zhou3}, ours is a one-phase greedy algorithm, where a most cost-effective star is selected in each iteration. The effectiveness of a star is measured by a delicately designed potential function.

There are two difficulties addressed. First, since the number of stars is exponential, identifying a most cost-effective star efficiently is challenging. We showed that under our potential function, a most cost-effective star has a special structure and thus can be found in polynomial time. Second, the potential function is not submodular, it eludes existing techniques used in submodular optimization. Although our previous works \cite{Shi1,Zhou2,Zhou1} successfully dealt with some cases of this problem for the {\em cardinality} version, those techniques cannot deal with the {\em weighted} version. A small error in the potential function makes the weight accumulate to an uncontrollable amount. In this paper, we proposed an amortized analysis, showing that although large errors are inevitable in some steps, they can be compensated overall. The crucial part is to establish an injective mapping from {\em fragments} of the computed solution to the fragments of the optimal solution so that such compensation is possible.

\section*{Acknowledgment}
This research is supported in part by National Natural Science Foundation of China (U20A2068) and NSF of USA under grant III-1907472.


\begin{thebibliography}{28}
\ifx \bisbn   \undefined \def \bisbn  #1{ISBN #1}\fi
\ifx \binits  \undefined \def \binits#1{#1}\fi
\ifx \bauthor  \undefined \def \bauthor#1{#1}\fi
\ifx \batitle  \undefined \def \batitle#1{#1}\fi
\ifx \bjtitle  \undefined \def \bjtitle#1{#1}\fi
\ifx \bvolume  \undefined \def \bvolume#1{\textbf{#1}}\fi
\ifx \byear  \undefined \def \byear#1{#1}\fi
\ifx \bissue  \undefined \def \bissue#1{#1}\fi
\ifx \bfpage  \undefined \def \bfpage#1{#1}\fi
\ifx \blpage  \undefined \def \blpage #1{#1}\fi
\ifx \burl  \undefined \def \burl#1{\textsf{#1}}\fi
\ifx \doiurl  \undefined \def \doiurl#1{\url{https://doi.org/#1}}\fi
\ifx \betal  \undefined \def \betal{\textit{et al.}}\fi
\ifx \binstitute  \undefined \def \binstitute#1{#1}\fi
\ifx \binstitutionaled  \undefined \def \binstitutionaled#1{#1}\fi
\ifx \bctitle  \undefined \def \bctitle#1{#1}\fi
\ifx \beditor  \undefined \def \beditor#1{#1}\fi
\ifx \bpublisher  \undefined \def \bpublisher#1{#1}\fi
\ifx \bbtitle  \undefined \def \bbtitle#1{#1}\fi
\ifx \bedition  \undefined \def \bedition#1{#1}\fi
\ifx \bseriesno  \undefined \def \bseriesno#1{#1}\fi
\ifx \blocation  \undefined \def \blocation#1{#1}\fi
\ifx \bsertitle  \undefined \def \bsertitle#1{#1}\fi
\ifx \bsnm \undefined \def \bsnm#1{#1}\fi
\ifx \bsuffix \undefined \def \bsuffix#1{#1}\fi
\ifx \bparticle \undefined \def \bparticle#1{#1}\fi
\ifx \barticle \undefined \def \barticle#1{#1}\fi
%\bibcommenthead
\ifx \bconfdate \undefined \def \bconfdate #1{#1}\fi
\ifx \botherref \undefined \def \botherref #1{#1}\fi
\ifx \url \undefined \def \url#1{\textsf{#1}}\fi
\ifx \bchapter \undefined \def \bchapter#1{#1}\fi
\ifx \bbook \undefined \def \bbook#1{#1}\fi
\ifx \bcomment \undefined \def \bcomment#1{#1}\fi
\ifx \oauthor \undefined \def \oauthor#1{#1}\fi
\ifx \citeauthoryear \undefined \def \citeauthoryear#1{#1}\fi
\ifx \endbibitem  \undefined \def \endbibitem {}\fi
\ifx \bconflocation  \undefined \def \bconflocation#1{#1}\fi
\ifx \arxivurl  \undefined \def \arxivurl#1{\textsf{#1}}\fi
\csname PreBibitemsHook\endcsname

%%% 1
\bibitem{Das}
\begin{bchapter}
\bauthor{\bsnm{Das}, \binits{B.}},
\bauthor{\bsnm{Bharghavan}, \binits{V.}}:
\bctitle{Routing in ad-hoc networks using minimum connected dominating sets}.
In: \bbtitle{ICC'97 - Proceedings of International Conference on
  Communications},
vol. \bseriesno{1},
pp. \bfpage{376}--\blpage{380}
(\byear{1997}).
\doiurl{10.1109/ICC.1997.605303}
\end{bchapter}
\endbibitem

%%% 2
\bibitem{Ephremides}
\begin{barticle}
\bauthor{\bsnm{Ephremides}, \binits{A.}},
\bauthor{\bsnm{Wieselthier}, \binits{J.E.}},
\bauthor{\bsnm{Baker}, \binits{D.J.}}:
\batitle{A design concept for reliable mobile radio networks with frequency
  hopping signaling}.
\bjtitle{Proceedings of the IEEE}
\bvolume{75}(\bissue{1}),
\bfpage{56}--\blpage{73}
(\byear{1987}).
\doiurl{10.1109/PROC.1987.13705}
\end{barticle}
\endbibitem

%%% 3
\bibitem{Dai}
\begin{barticle}
\bauthor{\bsnm{Dai}, \binits{F.}},
\bauthor{\bsnm{Wu}, \binits{J.}}:
\batitle{On constructing $k$-connected $k$-dominating set in wireless ad hoc
  and sensor networks}.
\bjtitle{Journal of Parallel and Distributed Computing}
\bvolume{66}(\bissue{7}),
\bfpage{947}--\blpage{958}
(\byear{2006}).
\doiurl{10.1016/j.jpdc.2005.12.010}
\end{barticle}
\endbibitem

%%% 4
\bibitem{Christoph}
\begin{bchapter}
\bauthor{\bsnm{Amb\"{u}hl}, \binits{C.}},
\bauthor{\bsnm{Erlebach}, \binits{T.}},
\bauthor{\bsnm{Mihal\'{a}k}, \binits{M.}},
\bauthor{\bsnm{Nunkesser}, \binits{M.}}:
\bctitle{Constant-factor approximation for minimum-weight (connected)
  dominating sets in unit disk graphs}.
In: \beditor{\bsnm{Broy}, \binits{M.}},
\beditor{\bsnm{Denert}, \binits{E.}} (eds.)
\bbtitle{APPROX'06/RANDOM'06},
pp. \bfpage{3}--\blpage{14}.
\bpublisher{Springer},
\blocation{Barcelona, Spain}
(\byear{2006}).
\doiurl{10.1007/11830924_3}
\end{bchapter}
\endbibitem

%%% 5
\bibitem{Erlebach}
\begin{bchapter}
\bauthor{\bsnm{Erlebach}, \binits{T.}},
\bauthor{\bsnm{Mihal\'{a}k}, \binits{M.}}:
\bctitle{A $(4+\varepsilon)$-approximation for the minimum-weight dominating
  set problem in unit disk graphs}.
In: \bbtitle{WAOA'09 - Proceedings of the 7th International Conference on
  Approximation and Online Algorithms},
pp. \bfpage{135}--\blpage{146}.
\bpublisher{Springer},
\blocation{Copenhagen, Denmark}
(\byear{2009}).
\doiurl{10.1007/978-3-642-12450-1_13}
\end{bchapter}
\endbibitem

%%% 6
\bibitem{Huang}
\begin{barticle}
\bauthor{\bsnm{Huang}, \binits{Y.}},
\bauthor{\bsnm{Gao}, \binits{X.}},
\bauthor{\bsnm{Zhang}, \binits{Z.}},
\bauthor{\bsnm{Wu}, \binits{W.}}:
\batitle{A better constant-factor approximation for weighted dominating set in
  unit disk graph}.
\bjtitle{Journal of Combinatorial Optimization}
\bvolume{18}(\bissue{2}),
\bfpage{179}--\blpage{194}
(\byear{2009}).
\doiurl{10.1007/s10878-008-9146-0}
\end{barticle}
\endbibitem

%%% 7
\bibitem{Li3}
\begin{bchapter}
\bauthor{\bsnm{Li}, \binits{J.}},
\bauthor{\bsnm{Jin}, \binits{Y.}}:
\bctitle{A ptas for the weighted unit disk cover problem}.
In: \beditor{\bsnm{Halld{\'{o}}rsson}, \binits{M.M.}},
\beditor{\bsnm{Iwama}, \binits{K.}},
\beditor{\bsnm{Kobayashi}, \binits{N.}},
\beditor{\bsnm{Speckmann}, \binits{B.}} (eds.)
\bbtitle{ICALP 2015 - Proceedings 42nd International Colloquium Automata,
  Languages, and Programming},
pp. \bfpage{898}--\blpage{909}.
\bpublisher{Springer},
\blocation{Kyoto, Japan}
(\byear{2015}).
\doiurl{10.1007/978-3-662-47672-7_73}
\end{bchapter}
\endbibitem

%%% 8
\bibitem{Guha1}
\begin{barticle}
\bauthor{\bsnm{Guha}, \binits{S.} \bsuffix{Sudiptoand~Khuller}}:
\batitle{Improved methods for approximating node weighted steiner trees and
  connected dominating sets}.
\bjtitle{Information and Computation}
\bvolume{150}(\bissue{1}),
\bfpage{57}--\blpage{74}
(\byear{1999}).
\doiurl{10.1006/inco.1998.2754}
\end{barticle}
\endbibitem

%%% 9
\bibitem{Zhou3}
\begin{barticle}
\bauthor{\bsnm{Zhou}, \binits{J.}},
\bauthor{\bsnm{Zhang}, \binits{Z.}},
\bauthor{\bsnm{Tang}, \binits{S.}},
\bauthor{\bsnm{Huang}, \binits{X.}},
\bauthor{\bsnm{Du}, \binits{D.-Z.}}:
\batitle{Breaking the $o(\ln n)$ barrier: An enhanced approximation algorithm for fault-tolerant minimum weight connected dominating set}.
\bjtitle{INFORMS Journal on Computing}
\bvolume{30}(\bissue{2}),
\bfpage{225}--\blpage{235}
(\byear{2018}).
\doiurl{10.1287/ijoc.2017.0775}
\end{barticle}
\endbibitem

%%% 10
\bibitem{Du1}
\begin{bbook}
\bauthor{\bsnm{Du}, \binits{D.-Z.}},
\bauthor{\bsnm{Ko}, \binits{K.-I.}},
\bauthor{\bsnm{Hu}, \binits{X.}}:
\bbtitle{Design and Analysis of Approximation Algorithms}.
\bsertitle{Springer Optimization and Its Applications, SOIA, volume 62}.
\bpublisher{Springer},
\blocation{New York}
(\byear{2012}).
\doiurl{10.1007/978-1-4614-1701-9}
\end{bbook}
\endbibitem

%%% 11
\bibitem{Wan1}
\begin{barticle}
\bauthor{\bsnm{Wan}, \binits{P.-J.}},
\bauthor{\bsnm{Alzoubi}, \binits{K.M.}},
\bauthor{\bsnm{Frieder}, \binits{O.}}:
\batitle{Distributed construction of connected dominating set in wireless ad
  hoc networks}.
\bjtitle{Mobile Networks and Applications}
\bvolume{9}(\bissue{2}),
\bfpage{141}--\blpage{149}
(\byear{2004}).
\doiurl{10.1023/B:MONE.0000013625.87793.13}
\end{barticle}
\endbibitem

%%% 12
\bibitem{Guha}
\begin{barticle}
\bauthor{\bsnm{Guha}, \binits{S.}},
\bauthor{\bsnm{Khuller}, \binits{S.}}:
\batitle{Approximation algorithms for connected dominating sets}.
\bjtitle{Algorithmica}
\bvolume{20}(\bissue{4}),
\bfpage{374}--\blpage{387}
(\byear{1998}).
\doiurl{10.1007/PL00009201}
\end{barticle}
\endbibitem

%%% 13
\bibitem{Ruan}
\begin{barticle}
\bauthor{\bsnm{Ruan}, \binits{L.}},
\bauthor{\bsnm{Du}, \binits{H.}},
\bauthor{\bsnm{Jia}, \binits{X.}},
\bauthor{\bsnm{Wu}, \binits{W.}},
\bauthor{\bsnm{Li}, \binits{Y.}},
\bauthor{\bsnm{Ko}, \binits{K.-I.}}:
\batitle{A greedy approximation for minimum connected dominating sets}.
\bjtitle{Theoretical Computer Science}
\bvolume{329}(\bissue{1-3}),
\bfpage{325}--\blpage{330}
(\byear{2004}).
\doiurl{10.1016/j.tcs.2004.08.013}
\end{barticle}
\endbibitem

%%% 14
\bibitem{Du}
\begin{bchapter}
\bauthor{\bsnm{Du}, \binits{D.-Z.}},
\bauthor{\bsnm{Graham}, \binits{R.L.}},
\bauthor{\bsnm{Pardalos}, \binits{P.M.}},
\bauthor{\bsnm{Wan}, \binits{P.-J.}},
\bauthor{\bsnm{Wu}, \binits{W.}},
\bauthor{\bsnm{Zhao}, \binits{W.}}:
\bctitle{Analysis of greedy approximations with nonsubmodular potential
  functions}.
In: \beditor{\bsnm{Teng}, \binits{S.-H.}} (ed.)
\bbtitle{SODA 2008 - Proceedings of the Nineteenth Annual ACM-SIAM Symposium on
  Discrete Algorithms},
pp. \bfpage{167}--\blpage{175}.
\bpublisher{SIAM},
\blocation{San Francisco, California, USA}
(\byear{2008}).
\doiurl{10.5555/1347082.1347101}
\end{bchapter}
\endbibitem

%%% 15
\bibitem{Cheng}
\begin{barticle}
\bauthor{\bsnm{Cheng}, \binits{X.}},
\bauthor{\bsnm{Huang}, \binits{X.}},
\bauthor{\bsnm{Li}, \binits{D.}},
\bauthor{\bsnm{Wu}, \binits{W.}},
\bauthor{\bsnm{Du}, \binits{D.-Z.}}:
\batitle{A polynomial-time approximation scheme for the minimum-connected
  dominating set in ad hoc wireless networks}.
\bjtitle{Networks}
\bvolume{42}(\bissue{4}),
\bfpage{202}--\blpage{208}
(\byear{2010}).
\doiurl{10.1002/net.10097}
\end{barticle}
\endbibitem

%%% 16
\bibitem{Zhang1}
\begin{barticle}
\bauthor{\bsnm{Zhang}, \binits{Z.}},
\bauthor{\bsnm{Gao}, \binits{X.}},
\bauthor{\bsnm{Wu}, \binits{W.}},
\bauthor{\bsnm{Du}, \binits{D.-Z.}}:
\batitle{A ptas for minimum connected dominating set in 3-dimensional wireless
  sensor networks}.
\bjtitle{Journal of Global Optimization}
\bvolume{45}(\bissue{3}),
\bfpage{451}--\blpage{458}
(\byear{2009}).
\doiurl{10.1007/s10898-008-9384-9}
\end{barticle}
\endbibitem

%%% 17
\bibitem{Shi1}
\begin{barticle}
\bauthor{\bsnm{Shi}, \binits{Y.}},
\bauthor{\bsnm{Zhang}, \binits{Y.}},
\bauthor{\bsnm{Zhang}, \binits{Z.}},
\bauthor{\bsnm{Wu}, \binits{W.}}:
\batitle{A greedy algorithm for the minimum 2-connected m-fold dominating set
  problem}.
\bjtitle{Journal of Combinatorial Optimization}
\bvolume{31}(\bissue{1}),
\bfpage{136}--\blpage{151}
(\byear{2016}).
\doiurl{10.1007/s10878-014-9720-6}
\end{barticle}
\endbibitem

%%% 18
\bibitem{Zhou2}
\begin{barticle}
\bauthor{\bsnm{Zhou}, \binits{J.}},
\bauthor{\bsnm{Zhang}, \binits{Z.}},
\bauthor{\bsnm{Tang}, \binits{S.}},
\bauthor{\bsnm{Huang}, \binits{X.}},
\bauthor{\bsnm{Mo}, \binits{Y.}},
\bauthor{\bsnm{Du}, \binits{D.-Z.}}:
\batitle{Fault-tolerant virtual backbone in heterogeneous wireless sensor
  network}.
\bjtitle{IEEE/ACM Transactions on Networking}
\bvolume{25}(\bissue{6}),
\bfpage{3487}--\blpage{3499}
(\byear{2017}).
\doiurl{10.1109/TNET.2017.2740328}
\end{barticle}
\endbibitem

%%% 19
\bibitem{Zhou1}
\begin{barticle}
\bauthor{\bsnm{Zhou}, \binits{J.}},
\bauthor{\bsnm{Zhang}, \binits{Z.}},
\bauthor{\bsnm{Wu}, \binits{W.}},
\bauthor{\bsnm{Xing}, \binits{K.}}:
\batitle{A greedy algorithm for the fault-tolerant connected dominating set in a general graph}.
\bjtitle{Journal of Combinatorial Optimization}
\bvolume{28}(\bissue{1}),
\bfpage{310}--\blpage{319}
(\byear{2014}).
\doiurl{10.1007/s10878-013-9638-4}
\end{barticle}
\endbibitem

%%% 20
\bibitem{Zhang}
\begin{barticle}
\bauthor{\bsnm{Zhang}, \binits{Z.}},
\bauthor{\bsnm{Zhou}, \binits{J.}},
\bauthor{\bsnm{Tang}, \binits{S.}},
\bauthor{\bsnm{Huang}, \binits{X.}},
\bauthor{\bsnm{Du}, \binits{D.-Z.}}:
\batitle{Computing minimum $k$-connected $m$-fold dominating set in general
  graphs}.
\bjtitle{INFORMS Journal on Computing}
\bvolume{30}(\bissue{2}),
\bfpage{217}--\blpage{224}
(\byear{2018}).
\doiurl{10.1287/ijoc.2017.0776}
\end{barticle}
\endbibitem

%%% 21
\bibitem{Shang}
\begin{barticle}
\bauthor{\bsnm{Shang}, \binits{W.}},
\bauthor{\bsnm{Yao}, \binits{F.}},
\bauthor{\bsnm{Wan}, \binits{P.}},
\bauthor{\bsnm{Hu}, \binits{X.}}:
\batitle{On minimum m-connected k-dominating set problem in unit disc graphs}.
\bjtitle{Journal of Combinatorial Optimization}
\bvolume{16}(\bissue{2}),
\bfpage{99}--\blpage{106}
(\byear{2008}).
\doiurl{10.1007/s10878-007-9124-y}
\end{barticle}
\endbibitem

%%% 22
\bibitem{Wang2}
\begin{barticle}
\bauthor{\bsnm{Wang}, \binits{W.}},
\bauthor{\bsnm{Kim}, \binits{D.}},
\bauthor{\bsnm{An}, \binits{M.K.}},
\bauthor{\bsnm{Gao}, \binits{W.}},
\bauthor{\bsnm{Li}, \binits{X.}},
\bauthor{\bsnm{Zhang}, \binits{Z.}},
\bauthor{\bsnm{Wu}, \binits{W.}}:
\batitle{On construction of quality fault-tolerant virtual backbone in wireless
  networks}.
\bjtitle{IEEE/ACM Transactions on Networking}
\bvolume{21}(\bissue{5}),
\bfpage{1499}--\blpage{1510}
(\byear{2013}).
\doiurl{10.1109/TNET.2012.2227791}
\end{barticle}
\endbibitem

%%% 23
\bibitem{Shi2}
\begin{barticle}
\bauthor{\bsnm{Shi}, \binits{Y.}},
\bauthor{\bsnm{Zhang}, \binits{Z.}},
\bauthor{\bsnm{Mo}, \binits{Y.}},
\bauthor{\bsnm{Du}, \binits{D.-Z.}}:
\batitle{Approximation algorithm for minimum weight fault-tolerant virtual
  backbone in unit disk graphs}.
\bjtitle{IEEE/ACM Transactions on Networking}
\bvolume{25}(\bissue{2}),
\bfpage{925}--\blpage{933}
(\byear{2017}).
\doiurl{10.1109/TNET.2016.2607723}
\end{barticle}
\endbibitem

%%% 24
\bibitem{Fukunaga}
\begin{barticle}
\bauthor{\bsnm{Fukunaga}, \binits{T.}}:
\batitle{Approximation algorithms for highly connected multi-dominating sets in unit disk graphs}.
\bjtitle{Algorithmica}
\bvolume{80}(\bissue{11}),
\bfpage{3270}--\blpage{3292}
(\byear{2018}).
\doiurl{10.1007/s00453-017-0385-2}
\end{barticle}
\endbibitem

%%% 25
\bibitem{Nutov2}
\begin{barticle}
\bauthor{\bsnm{Nutov}, \binits{Z.}}:
\batitle{Approximating $k$-connected $m$-dominating set problems}.
\bjtitle{Algorithmica}
\bvolume{84}(\bissue{6}),
\bfpage{1511}--\blpage{1525}
(\byear{2022}).
\doiurl{10.1007/s00453-022-00935-x}
\end{barticle}
\endbibitem

%%% 26
\bibitem{Zou}
\begin{barticle}
\bauthor{\bsnm{Zou}, \binits{F.}},
\bauthor{\bsnm{Wang}, \binits{Y.}},
\bauthor{\bsnm{Xu}, \binits{X.-H.}},
\bauthor{\bsnm{Li}, \binits{X.}},
\bauthor{\bsnm{Du}, \binits{H.}},
\bauthor{\bsnm{Wan}, \binits{P.}},
\bauthor{\bsnm{Wu}, \binits{W.}}:
\batitle{New approximations for minimum-weighted dominating sets and
  minimum-weighted connected dominating sets on unit disk graphs}.
\bjtitle{Theoretical Computer Science}
\bvolume{412}(\bissue{3}),
\bfpage{198}--\blpage{208}
(\byear{2011}).
\doiurl{10.1016/j.tcs.2009.06.022}
\end{barticle}
\endbibitem

%%% 27
\bibitem{DuWan}
\begin{bbook}
\bauthor{\bsnm{Du}, \binits{D.-Z.}},
\bauthor{\bsnm{Wan}, \binits{P.-J.}}:
\bbtitle{Connected Dominating Set: Theory and Applications}.
\bsertitle{Springer Optimization and Its Applications, SOIA, volume 77}.
\bpublisher{Springer},
\blocation{New York}
(\byear{2012}).
\doiurl{10.1007/978-1-4614-5242-3}
\end{bbook}
\endbibitem

%%% 28
\bibitem{WZLD}
\begin{bbook}
\bauthor{\bsnm{Wu}, \binits{W.}},
\bauthor{\bsnm{Zhang}, \binits{Z.}},
\bauthor{\bsnm{Lee}, \binits{W.}},
\bauthor{\bsnm{Du}, \binits{D.-Z.}}:
\bbtitle{Optimal Coverage in Wireless Sensor Networks}.
\bsertitle{Springer Optimization and Its Applications, SOIA, volume 162}.
\bpublisher{Springer},
\blocation{New York}
(\byear{2020}).
\doiurl{10.1007/978-3-030-52824-9}
\end{bbook}
\endbibitem

%%% 29
\bibitem{ZHANG3}
\begin{barticle}
\bauthor{\bsnm{Zhang}, \binits{Z.}}:
\batitle{Survey of approximation algorithm on virtual backbone of wireless sensor network (chinese)}.
\bjtitle{Journal of Computer Research and Development}
\bvolume{53}(\bissue{1}),
\bfpage{15}--\blpage{25}
(\byear{2016}).
\doiurl{10.7544/issn1000-1239.2016.20.2015065}
\end{barticle}
\endbibitem

\end{thebibliography}
\end{document}